\documentclass[DIV9,a4paper,final]{scrartcl}

\usepackage[utf8]{inputenc}
\usepackage[T1]{fontenc}

\usepackage{tabularx}
\usepackage{amssymb}
\usepackage{amsthm}
\usepackage{lmodern}

\usepackage{microtype}
\usepackage{nccmath}
\usepackage{algorithm}
\usepackage{algpseudocode}
\usepackage[table]{xcolor}
\usepackage{wrapfig}
\usepackage{tikz}

\algnewcommand{\LineIf}[2]{\State \algorithmicif\, #1 \,\algorithmicthen\, #2 \,\algorithmicend\ \algorithmicif}
\algnewcommand{\LineIfElse}[3]{\State \algorithmicif\, #1 \,\algorithmicthen\, #2 \,\algorithmicelse\, #3 \,\algorithmicend\ \algorithmicif}
\algnewcommand{\LineForAll}[2]{\State \algorithmicforall\, #1 \,\algorithmicdo\, #2 \,\algorithmicend\ \algorithmicfor}

\renewcommand{\leq}{\leqslant}

\renewcommand{\le}{\leq}

\newcommand{\union}{\mathbin{\cup}}
\newcommand{\bigunion}{\mathop{\bigcup}}
\newcommand{\intersect}{\mathbin{\cap}}
\newcommand{\MonInt}{\ensuremath{\textsc{MonIsect}}}
\newcommand{\Memb}{\ensuremath{\textsc{Memb}}}
\newcommand{\PS}{\ensuremath{\textsc{DfaIsect}}}

\newcommand{\abs}[1] {\ensuremath\left|#1\right|}
\newcommand{\set}[2] {\ensuremath{\left\{#1 \mid #2\right\}}}
\newcommand{\os}[1] {\ensuremath{\left\{#1\right\}}}
\newcommand{\Jeq} {\mathrel{\ensuremath{\mathcal{J}}}}
\newcommand{\Jle} {\mathrel{\ensuremath{\le_\mathcal{J}}}}

\newcommand{\NC}{\ensuremath{\mathsf{NC}}}
\newcommand{\PTIME}{\ensuremath{\mathsf{P}}}
\newcommand{\NP}{\ensuremath{\mathsf{NP}}}
\newcommand{\PSPACE}{\ensuremath{\mathsf{PSPACE}}}
\newcommand{\NPSPACE}{\ensuremath{\mathsf{NPSPACE}}}
\newcommand{\vV}{\mathbf{V}}
\newcommand{\vTrivial}{\mathbf{1}}
\newcommand{\vAp}{\mathbf{A}}
\newcommand{\vMon}{\mathbf{Mon}}
\newcommand{\vDA}{\mathbf{DA}}
\newcommand{\vDS}{\mathbf{DS}}
\newcommand{\vDO}{\mathbf{DO}}
\newcommand{\vG}{\mathbf{G}}
\newcommand{\vRone}{\mathbf{R_1}}
\newcommand{\vLone}{\mathbf{L_1}}
\newcommand{\vAone}{\mathbf{A_1}}
\newcommand{\vR}{\mathbf{R}}
\newcommand{\vAcomTwo}{\mathbf{Acom_2}}
\newcommand{\vXR}{\mathbf{XR}}
\newcommand{\vXL}{\mathbf{XL}}
\newcommand{\Sift}{\textsc{sift}}
\newcommand{\Init}{\textsc{init}}
\newcommand{\alp}{\mathsf{alph}}
\newcommand{\val}{\mathsf{val}}
\newcommand{\abstar}{\ensuremath{B_2^1}}
\DeclareMathOperator{\lcm}{lcm}

\newcommand{\newres}{\cellcolor[rgb]{0.99,0.78,0.07}}
\newcommand{\redres}{\cellcolor[rgb]{0.99,0.89,0.54}}

\newcommand{\N}{\mathbb{N}}
\newcommand{\R}{\mathbb{R}}
\newcommand{\T}{\mathcal{T}}
\newcommand{\lbl}{\Lambda}

\newcommand{\ie}{i.e.,~}
\newcommand{\eg}{e.g.~}
\newcommand{\ms}{\hspace*{0.5pt}}

\makeatletter
\newcommand{\thickhline}{%
    \noalign {\ifnum 0=`}\fi \hrule height 1pt
    \futurelet \reserved@a \@xhline
}
\newcolumntype{"}{@{\hskip\tabcolsep\vrule width 1pt\hskip\tabcolsep}}
\newcolumntype{Y}{>{\centering\arraybackslash}X}

\makeatother

\newtheorem{theorem}{Theorem}
\newtheorem{lemma}[theorem]{Lemma}
\newtheorem{corollary}[theorem]{Corollary}
\newtheorem{proposition}[theorem]{Proposition}

\title{The Intersection Problem \\ for Finite Monoids}

\author{Lukas Fleischer \and Manfred Kuf\-leitner}
\date{FMI, University of Stuttgart\thanks{This work was supported by the DFG grants DI 435/5-2 and \mbox{KU 2716/1-1}.}\\[.1mm]
  Universitätsstraße 38, 70569 Stuttgart, Germany\\
  \texttt{$\{$fleischer,kufleitner$\}$@fmi.uni-stuttgart.de}}

\begin{document}

\maketitle

\begin{abstract}
  We investigate the intersection problem for finite monoids, which asks for a
  given set of regular languages, represented by recognizing morphisms to
  finite monoids from a variety $\vV$, whether there exists a word contained in
  their intersection.
  Our main result is that the problem is $\PSPACE$-complete if $\vV
  \not\subseteq \vDS$ and $\NP$-complete if $\vTrivial \subsetneq \vV \subseteq
  \vDO$.
  Our $\NP$-algorithm for the case $\vV \subseteq \vDO$ uses novel methods,
  based on compression techniques and combinatorial properties of $\vDO$.
  We also show that the problem is log-space reducible to the intersection
  problem for deterministic finite automata (DFA) and that a variant of the
  problem is log-space reducible to the membership problem for transformation
  monoids. In light of these reductions, our hardness results can be seen as a
  generalization of both a classical result by Kozen~\cite{koz77} and a theorem
  by Beaudry, McKenzie and Th\'{e}rien~\cite{BeaudryMT92}.
\end{abstract}

\section{Introduction}

In 1977, Kozen showed that deciding whether the intersection of the languages
recognized by a set of given deterministic finite automata (DFA) is non-empty
is $\PSPACE$-complete~\cite{koz77}. This result has since been the building
block for numerous hardness results in formal language theory and related
fields; see \eg\cite{Bernatsky1997,chohuynh91,dgh05IC,JiangR91}.
It is natural to ask whether the problem becomes easier when restricting the
input. Various special cases, such as bounding the number $k$ of automata in
the input~\cite{LangeR92} or considering only automata with a fixed number of
accepting states~\cite{BlondinKM16}, were investigated in follow-up work;
see~\cite{HolzerK11} for a survey.

Another very natural restriction is to only consider automata with certain
\emph{structural properties}.
One such property is \emph{counter-freeness}: an automaton is
\emph{counter-free} if no word permutes a non-trivial subset of its states.  By
a famous result of Schützenberger~\cite{sch65sf}, a regular language is
recognized by a counter-free automaton if and only if it is star-free.
These properties are often expressed using the algebraic framework:
instead of considering the automaton itself, one considers its transition
monoid. The latter is the transformation monoid generated by the action of the letters
on the set of states. Now, properties of automata can be given by membership of the transition
monoid in certain classes, so-called \emph{varieties}, of finite monoids.
For example, an automaton is counter-free if and only if its transition monoids
belongs to the variety $\vAp$ of \emph{aperiodic} monoids.
The DFA intersection problem for a variety $\vV$, denoted by $\PS(\vV)$, is
formalized as follows.
\vspace{1em}

\noindent
\begin{tabularx}{\textwidth}{p{1.15cm}X}
  \thickhline
  $\PS(\vV)$ \\
  \hline
  \textsf{Input}: & DFAs $A_1, \dots, A_k$ with transition monoids from $\vV$ \\
  \textsf{Question}: & Is $L(A_1) \intersect \cdots \intersect L(A_k) \ne \emptyset$? \\
  \hline
\end{tabularx}
\vspace{1ex}

Note that $\PS(\vMon)$, where $\vMon$ is the variety of all finite monoids, is
the general DFA intersection problem considered by Kozen.
A careful inspection of his proof actually reveals that $\PS(\vAp)$ is
$\PSPACE$-complete already~\cite{chohuynh91}.
Additionally requiring all DFAs to have a single accepting state, we obtain a
variant of $\PS(\vV)$ reminiscent of another problem investigated by Kozen, the
\emph{membership problem for transformation monoids}.
\vspace{1em}

\noindent
\begin{tabularx}{\textwidth}{p{1.15cm}X}
  \thickhline
  $\Memb(\vV)$ \\
  \hline
  \textsf{Input}: & Transformations $f_1, \dots, f_m \colon X \to X$ generating a monoid $T \in \vV$ and $g \colon X \to X$ \\
  \textsf{Question}: & Does $g$ belong to $T$? \\
  \hline
\end{tabularx}
\vspace{1ex}

The complexity of $\Memb(\vV)$ was studied extensively in a series of 
papers~\cite{BabaiLS87,Beaudry88,Beaudry88thesis,Beaudry94,BeaudryMT92,FurstHopcroftLuks80,Sims1967}.
However, for certain varieties $\vV$, obtaining the exact complexity of
$\PS(\vV)$ and $\Memb(\vV)$ is a challenging problem.
To date, only partial results are known, see Table~\ref{tab:summary}. For
example, it is open whether or not $\Memb(\vDA) \in \NP$, a question stated
explicitly in~\cite{BeaudryMT92} and revisited in~\cite{tt02} around ten years
later.

\begin{table}
  \renewcommand{\arraystretch}{1.15}
  \centering
  \begin{tabularx}{\textwidth}{|Y"Y|Y|Y|}
    \hline
    & $\MonInt(\vV)$ & $\MonInt_1(\vV)$ & $\Memb(\vV)$~\cite{BabaiLS87,BeaudryMT92} \\
    \thickhline
    $\NC$ & --- & \redres{$\vV \subseteq \vG$} & $\vV \subseteq \vG$ \\
    $\PTIME$ & --- & \redres{$\vV \subseteq \vRone \lor \vLone$} & $\vV \subseteq \vRone \lor \vLone$ \\
    $\NP$ & \newres{$\vV \subseteq \vDO$} & \redres{$\vV \subseteq \vDO$} & $\vV \subseteq \vR, \vV \subseteq \vAone$ \\[1mm]
    $\NP$-hard & \newres{all $\vV \ne \vTrivial$} & --- & {$\vAcomTwo \subseteq \vV$,\- $\vXR \subseteq \vV, \vXL \subseteq \vV$} \\
    $\PSPACE$ & \redres{all $\vV$} & \redres{all $\vV$} & all $\vV$ \\
    $\PSPACE$-hard & \redres{$\vV \not\subseteq \vDS$} & \newres{$\vV \not\subseteq \vDS$} & $\vV \not\subseteq \vDS$ \\
    \hline
  \end{tabularx}
  \caption{Summary of complexity results ({\textcolor[rgb]{0.99,0.78,0.07}{\rule{5mm}{2mm}}} new main result, {\textcolor[rgb]{0.99,0.89,0.54}{\rule{5mm}{2mm}}} follows from reductions)
  }
  \label{tab:summary}
  \renewcommand{\arraystretch}{1.0}
\end{table}

Since algebraic tools are already used to express structural properties of
automata, it seems natural to consider the fully algebraic version of the
intersection problem by directly using finite monoids as language acceptors
instead of taking the detour via automata and their transition monoids.
A language $L \subseteq A^*$ is \emph{recognized} by a morphism $h \colon A^*
\to M$ to a finite monoid $M$ if $L = h^{-1}(P)$ for some subset $P$ of $M$.
The set $P$ is often called the \emph{accepting set} because it resembles the accepting
states in finite automata.
A monoid $M$ recognizes a language $L \subseteq A^*$ if there exists a morphism
$h \colon A^* \to M$ recognizing $L$.
It is well-known that a language is recognized by a finite monoid if and only
if it is regular.
For a variety of finite monoids $\vV$, the \emph{intersection problem for
$\vV$} is defined as follows.
\vspace{1em}

\noindent
\begin{tabularx}{\textwidth}{p{1.15cm}X}
  \thickhline
  $\MonInt(\vV)$ \\
  \hline
  \textsf{Input}: & Morphisms $h_i \colon A^* \to M_i \in \vV$ and sets $P_i \subseteq M_i$ with $1 \le i \le k$ \\
  \textsf{Question}: & Is $h_1^{-1}(P_1) \intersect \cdots \intersect h_k^{-1}(P_k) \ne \emptyset$? \\
  \hline
\end{tabularx}
\vspace{1ex}

We assume that the monoids are given as multiplication tables, such that,
assuming a random-access machine model, multiplications can be performed in
logarithmic time.

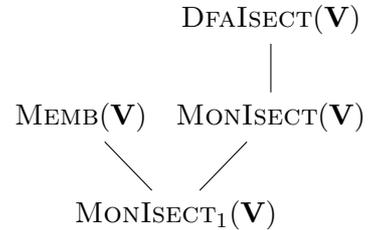
\begin{wrapfigure}{r}{0.35\textwidth}
  \centering
  \begin{tikzpicture}[grow=up,level distance=1.3cm,sibling distance=2.5cm,every node/.style={align=center}]
    \node {$\MonInt_1(\vV)$} child { node {$\MonInt(\vV)$} child { node {$\PS(\vV)$} }  } child { node {$\Memb(\vV)$} };
  \end{tikzpicture}
  \caption{Relations between the problems considered in this work}
  \label{fig:reductions}
\end{wrapfigure}

There is a close connection to both the DFA intersection problem and the
membership problem for transformation monoids.
More specifically, for every variety~$\vV$, there is a log-space reduction of
$\MonInt(\vV)$ to $\PS(\vV)$.
The variant $\MonInt_1(\vV)$ of the finite monoid intersection problem, where
each of the accepting sets is a singleton, can be reduced to $\Memb(\vV)$.
Our reducibility results are depicted in Figure~\ref{fig:reductions}.

Not only is the algebraic version of the intersection problem a natural problem
to consider, making progress in classifying its complexity also raises hope to
make progress in solving open complexity questions regarding
$\PS(\vV)$ and $\Memb(\vV)$. Using novel techniques, we prove that
$\MonInt(\vV)$ is $\NP$-complete whenever $\vV \subseteq \vDO$ and
$\PSPACE$-complete whenever $\vV \not\subseteq \vDS$.
In particular, since $\vDA$ is a subset of $\vDO$, we obtain an $\NP$-algorithm
for $\MonInt(\vDA)$ while the problem of whether there exists such an algorithm
for $\Memb(\vDA)$ or $\PS(\vDA)$ has been open for more than 25~years.
Moreover, in view of the reductions mentioned above, our $\PSPACE$-hardness
result can be seen as a generalization of both Kozen's result and a result from
\cite{BeaudryMT92}, stating that every variety of aperiodic monoids not
contained within $\vDA = \vDS \intersect \vAp$ admits a $\PSPACE$-complete
transformation monoid membership problem.

Our results are summarized in Table~\ref{tab:summary}. Only a very small gap of
varieties contained within $\vDS$ but not $\vDO$ remains.
Answering complexity questions in this setting is deeply connected to
understanding the languages recognized by monoids in $\vDS$ which is another
problem open for over twenty years; see \eg\cite[Open Problem 14]{alm94}.
Obtaining a dichotomy result for $\MonInt(\vV)$ is likely to provide new major
insights for both $\PS(\vV)$ and the language variety corresponding to $\vDS$,
and, conversely, new insights on either language properties of $\vDS$ or on
$\PS(\vDS)$ will potentially help with obtaining such a result.

We conclude with a first complexity result on the intersection problem for
finite monoids.

\begin{theorem}
  $\MonInt(\vMon) \in \PSPACE$.
  \label{thm:pspace}
\end{theorem}
\begin{proof}
  Since $\PSPACE$ = $\NPSPACE$ by Savitch's Theorem, it suffices to give a
  non-deterministic algorithm which requires polynomial space.
  The algorithm proceeds by guessing a word in the intersection, letter by
  letter. The word is not written down explicitly but after each guess, the
  image of the current prefix under each morphism is computed and stored.
  Finally, the algorithm verifies that each of the images is in the
  corresponding accepting set.
\end{proof}

\section{Preliminaries}
\label{sec:prelim}

\subparagraph{Words and Languages.}

Let $A$ be a finite alphabet.
A \emph{word} over $A$ is a finite sequence of letters $a_1 \cdots a_\ell$
with $a_i \in A$ for all $i \in \os{1, \dots, \ell}$.
The set $A^*$ denotes the set of all words over $A$ and a \emph{language} is a
subset of $A^*$.
The \emph{content} (or \emph{alphabet}) of a word $w = a_1 \cdots a_\ell$ is the subset
$\alp(w) = \os{a_1, \dots, a_\ell}$ of $A$.
A word $u$ is a \emph{factor} of $w$ if there exist $p, q \in A^*$
such that $w = puq$; and, when the factorization is fixed, then the position of
$u$ is called its \emph{occurrence}.

\subparagraph{Algebra.}

Let $M$ be a finite monoid. An element $e \in M$ is \emph{idempotent} if $e^2 =
e$. The set of all idempotent elements of $M$ is denoted by $E(M)$. In a finite
monoid $M$, the integer $\omega_M = \abs{M}!$ plays an important role: for
each $m \in M$, the element $m^{\omega_M}$ is idempotent. For convenience, we
often write $\omega$ instead of $\omega_M$ if the reference to $M$ is clear
from the context.
For two elements $m, n \in M$, we write $m \Jle n$ if the two-sided ideal of
$m$ is contained in the two-sided ideal of $n$, \ie{}$MmM \subseteq MnM$.
We write $m \Jeq n$ if both $m \Jle n$ and $n \Jle m$.

The \emph{direct product} of two monoids $M$ and $N$ is the Cartesian product
$M \times N$ with componentwise multiplication. 
A monoid $N$ is a \emph{quotient} of a monoid $M$ if there exists a surjective
morphism $h \colon M \to N$.
A monoid $N$ is a \emph{divisor} of a monoid $M$ if $N$ is a quotient of a
submonoid of $M$.

A \emph{variety} of finite monoids is a class $\vV$ of finite monoids which is
closed under (finite) direct products and divisors. The class of all finite
monoids $\vMon$ is a variety. The following other varieties play an important
role in this work:

\begin{align*}
  \vG & = \set{M \in \vMon}{\forall e \in E(M): e = 1} \\
  \vDS & = \set{M \in \vMon}{\forall e, f \in E(M): e \Jeq f \implies (efe)^\omega = e} \\
  \vDO & = \set{M \in \vMon}{\forall e, f \in E(M): e \Jeq f \implies efe = e}
\end{align*}

It is easy to see that $\vG$ contains exactly those finite monoids which are
groups. Since direct products of groups are groups and divisors of groups are
groups, $\vG$ is indeed a variety.
For proofs that $\vDS$ and $\vDO$ are varieties, we refer to~\cite{pin86}.
From the definitions, it follows immediately that $\vDO \subseteq \vDS$.
There exist several other interesting characterizations of $\vDS$. Let
$\abstar$ be the monoid defined on the set $\os{1, a, b, ab, ba, 0}$ by the
operation $aba = a$, $bab = b$ and $a^2 = b^2 = 0$ where $0$ is a zero element.
Then the following holds, see \eg\cite{alm94}.

\begin{proposition}
  Let $M$ be a finite monoid. The following properties are equivalent:
  \begin{enumerate}
    \item $M \in \vDS$.
    \item For each $e \in E(M)$ and $x \in M$ with $e \Jle x$, we have $(exe)^\omega = e$.
    \item For each $e \in E(M)$, the elements $\set{x \in M}{e \Jle x}$ form a submonoid of $M$.
    \item $\abstar$ is not a divisor of $M \times M$.
  \end{enumerate}
  \label{prop:ds-char}
\end{proposition}

\subparagraph{Tiling Systems.}

A \emph{tiling system} is a tuple $\T = (\lbl, T, n, f, b)$ where $\lbl$ is a
finite set of \emph{labels}, $T \subseteq \lbl \times \lbl \times \lbl \times
\lbl$ are the so-called \emph{tiles}, $n \in \N$ is the \emph{width} and $f, b
\in T^n$ are the \emph{first row} and \emph{bottom row}.
For a tile $t=(t_w, t_e, t_s, t_n) \in T$, we let $\lambda_w(t) = t_w$,
$\lambda_e(t) = t_e$, $\lambda_s(t) = t_s$ and $\lambda_n(t) = t_n$. These
labels can be thought of as labels in \emph{west}, \emph{east}, \emph{south}
and \emph{north} direction.
An \emph{$m$-tiling} of $\T$ is a mapping $\tau \colon \os{1, \dots, m} \times
\os{1, \dots, n} \to T$ such that the following properties hold:

\begin{enumerate}
  \item $\tau(1, 1) \tau(1, 2) \cdots \tau(1, n) = f$,
  \item $\lambda_e(\tau(i, j))  = \lambda_w(\tau(i, j+1))$ for $1 \le i \le m$ and $1 \le j \le n-1$,
  \item $\lambda_s(\tau(i, j))  = \lambda_n(\tau(i+1, j))$ for $1 \le i \le m-1$ and $1 \le j \le n$,
  \item $\tau(m, 1) \tau(m, 2) \cdots \tau(m, n) = b$.
\end{enumerate}

The \emph{corridor tiling problem} asks for a given tiling system $\T$ whether
there exists some $m \in \N$ such that there is a $m$-tiling of $\T$.
The \emph{square tiling problem} asks for a given tiling system $\T$ of width
$n$, whether there exists an $n$-tiling of $\T$.
It is well-known that the corridor tiling problem is $\PSPACE$-complete and
that the square tiling problem is $\NP$-complete~\cite{Chlebus86}.

\subparagraph{Straight-Line Programs.}

A \emph{straight-line program} (\emph{SLP}) is a grammar $S = (V, A, P, X_s)$
where $V$ is a finite set of variables, $A$ is a finite alphabet, $P \colon V
\to (V \union A)^*$ is a mapping and $X_s \in V$ is the so-called \emph{start
variable}.
For a variable $X \in V$, the word $P(X)$ is the \emph{right-hand side} of $X$.
We require that there exists a linear order $<$ on $V$ such that $Y < X$
whenever $P(X) \in (V \union A)^* Y (V \union A)^*$.
Starting with some word $\alpha \in (V \union A)^*$ and repeatedly replacing
variables $X \in V$ by $P(X)$ yields a word from $A^*$, the so called
\emph{evaluation of $\alpha$}, denoted by $\val(\alpha)$.
The word \emph{produced by $S$} is $\val(S) = \val(X_s)$.
If the reference to $A$ and $V$ is clear from the context, we will often use
the notation $h(\alpha)$ instead of $h(\val(\alpha))$ for the image of the
evaluation of a word $\alpha \in (A \union V)^*$ under a morphism $h \colon A^*
\to M$. Analogously, we write $h(S)$ instead of $h(\val(S))$.
The \emph{size} of $S$ is $\abs{S} = \sum_{X \in V} \abs{P(X)}$.
Each variable $X$ of an SLP $S$ can be viewed as an SLP itself by making $X$
the start variable of $S$.

The following simple lemma illustrates how SLPs can be used for compression.

\begin{lemma}
  Let $S = (V, A, P, X_s)$ be an SLP and let $e \in \N$.
  Let $w$ be the word produced by~$S$.
  Then there exists an SLP $S'$ of size $\abs{S'} \le \abs{S} + 4 \log(e)$ such
  that $S'$ produces $w^e$.
  \label{lem:slp}
\end{lemma}
\begin{proof}
  We obtain $S'$ by iteratively adding new variables to $V$ as follows,
  starting with $i = e$ and repeating the process until $i = 0$.
  \begin{itemize}
    \item If $i >0$ is odd, add a new variable $X_i$ and let $P(X_i) = X_{i-1}
      X_s$. Let $i := i-1$.
    \item If $i >0$ is even, add a new variable $X_i$ and let $P(X_i) = X_{i/2}
      X_{i/2}$. Let $i := i/2$.
  \end{itemize}
  Finally, add the variable $X_0$ and let $P(X_0) = \varepsilon$.
  The new start variable is $X_e$ and by construction, we have $\val(X_e) =
  w^e$.
\end{proof}

\section{Connections to Other Problems}
\label{sec:connections}

Before investigating the complexity of $\MonInt(\vV)$ itself, we establish
connections to other well-known problems defined in the introduction, starting
with the DFA intersection problem.

\begin{proposition}
  Let $\vV$ be a variety of finite monoids, let $M \in \vV$, let $h \colon A^*
  \to M$ be a morphism and let $P \subseteq M$.
  Then there exists a finite deterministic automaton $A$ with $\abs{M}$ states
  such that $L(A) = h^{-1}(P)$ and such that the transition monoid of $A$
  belongs to $\vV$.
  When the monoid, the morphism and the accepting set are given as inputs, this
  automaton is log-space computable.
  \label{prop:ps}
\end{proposition}
\begin{proof}
  It suffices to perform the standard conversion of monoids to finite automata.
  The set of states of $A$ is $M$, the initial state is the identity element
  $1$, the transitions are defined by $\delta(m, a) = m \ms h(a)$ for all $m
  \in M$ and $a \in A$ and the accepting states are $P$.
  A straightforward verification shows that the transition monoid of $A$ is
  isomorphic to $M$.
  Since computing images $h(a)$ and performing multiplications are just table
  lookups, each output bit can be computed in logarithmic time on a
  random-access machine model.
\end{proof}

\begin{corollary}
  For each variety of finite monoids $\vV$, the problem $\MonInt(\vV)$ is
  log-space reducible to $\PS(\vV)$.
  \label{crl:ps}
\end{corollary}

For a direct link to $\Memb(\vV)$, we consider the variant $\MonInt_1(\vV)$ of
the finite monoid intersection problem. In this variant, each of the accepting
sets is a singleton.

\begin{proposition}
  Let $\vV$ be a variety of finite monoids and let $M_1, \ldots, M_k \in
  \vV$ be pairwise disjoint finite monoids. For each $i \in \os{1, \dots, k}$,
  let $h_i: A^* \to M_i$ be a morphism and let $p_i \in M_i$.
  Then there exists a transformation monoid $T \in \vV$ on the set $M = M_1
   \union \cdots \union M_k$, a morphism $h \colon A^* \to T$ and a
  transformation $p \in T$ such that $h^{-1}(p) = h_1^{-1}(p_1) \intersect \cdots \intersect h_k^{-1}(p_k)$.
  %
  %Moreover, the generators of $T$, the morphism $h$ and the transformation $p$
  %are log-space computable.
  \label{prop:memb}
\end{proposition}
\begin{proof}
  For each $a \in A$, we define a transformation $f_a : M \to M$ by 
  $f_a(m) = m \ms h_i(a)$ for $m \in
  M_i$.
  The closure of $\set{f_a}{a \in A}$ under composition is the transformation
  monoid $T$ and the morphism $h \colon A^* \to T$ is given by $h(a) = f_a$. We
  let $p \colon M \to M$ be the transformation defined by $p(m) = m \ms p_i$
  for $m \in M_i$.
  \pagebreak[3]

  We need to verify that $h^{-1}(p) = h_1^{-1}(p_1) 
  \intersect \cdots \intersect h_k^{-1}(p_k)$.
  For the inclusion from right to left, let $w \in A^*$ be a word such that
  $h_i(w) = p_i$ for each $i \in \os{1, \dots, k}$. Then, by definition, $h(w)$
  is the transformation which maps an element $m \in M_i$ to $m \ms h_i(w) = m
  \ms p_i$, \ie{}$h(w) = p$.
  The converse inclusion is trivial.

  It is easy to check that $T$ is a divisor of the direct product $M_1  \times \cdots \times M_k$ and thus, by closure of $\vV$ under direct
  products and under division, $T$ belongs to $\vV$ as well.
  Since computing images $h_i(a)$ and performing multiplications are just table
  lookups, each output bit can be computed in logarithmic time on a
  random-access machine model.
\end{proof}

\begin{corollary}
  For each variety of finite monoids $\vV$, the problem $\MonInt_1(\vV)$ is
  log-space reducible to $\Memb(\vV)$.
  \label{crl:memb}
\end{corollary}

\section{Hardness Results}
\label{sec:hardness}

The following lower bound can be viewed as a variant of classical
$\NP$-hardness results and is based on the well-known fact that each
non-trivial variety contains either the monoid $U_1 = \os{0, 1}$ with integer
multiplication or a finite cyclic group (however, the proof itself does not
require this case distinction).

\begin{theorem}
  Let $\vV$ be a non-trivial variety of finite monoids. Then, the decision
  problem $\MonInt(\vV)$ is $\NP$-hard.
  \label{thm:np-hard}
\end{theorem}
\begin{proof}
  We give a polynomial-time reduction of the square tiling problem to
  $\MonInt(\vV)$.

  Let $\T = (\lbl, T, n, f, b)$ be a tiling system.
  Let $M \in \vV$ be a non-trivial finite monoid and let $x \in M \setminus
  \os{1}$.
  The alphabet $A$ is the set $T \times \os{1, \dots, n} \times \os{1, 
  \dots, n}$.
  Let $f = t_1 \cdots t_n$. For each integer $j \in \os{1, \dots, n}$
  and each direction $d \in \os{w, e, s, n}$, we define a morphism $f_{j,d}
  \colon A \to M$ by mapping $(t, 1, j)$ to $x$ if $\lambda_d(t) =
  \lambda_d(t_j)$ and mapping the remaining letters to~$1$.
  Analogously, with $b = u_1 \cdots u_n$, we let $b_{j,d} \colon A \to M$
  be the morphism mapping $(t, n, j)$ to $x$ if $\lambda_d(t) = \lambda_d(u_j)$
  and mapping other letters to~$1$.
  For each integer $i \in \os{1, \dots, n}$, each $j \in \os{1, \dots,
  n-1}$ and each label $\mu \in \lbl$, we define a morphism $h_{i,j,\mu} \colon
  A \to M \times M$ by
  \begin{equation*}
    h_{i,j,\mu}(t, k, \ell) = 
    \begin{cases}
      (x, 1) & \text{if $k = i$, $\ell = j$ and $\lambda_e(t) = \mu$} \\
      (1, x) & \text{if $k = i$, $\ell = j+1$ and $\lambda_w(t) = \mu$} \\
      (1, 1) & \text{otherwise}
    \end{cases}
  \end{equation*}
  and, analogously, we define morphisms $v_{i,j,\mu} \colon A \to M \times M$
  with $i \in \os{1, \dots, n-1}$ and $j \in \os{1, \dots, n}$ and $\mu
  \in \lbl$ as follows:
  \begin{equation*}
    v_{i,j,\mu}(t, k, \ell) = 
    \begin{cases}
      (x, 1) & \text{if $k = i$, $\ell = j$ and $\lambda_s(t) = \mu$} \\
      (1, x) & \text{if $k = i+1$, $\ell = j$ and $\lambda_n(t) = \mu$} \\
      (1, 1) & \text{otherwise}
    \end{cases}
  \end{equation*}
  Finally, we define morphisms $g_{i,j,d,\mu,\mu'} \colon A \to M \times M$
  with $i,j  \in \os{1, \dots, n}$, $d \in \os{w, e, s, n}$ as well as
  $\mu, \mu' \in \lbl$ and $\mu \ne \mu'$ as follows:
  \begin{equation*}
    g_{i,j,d,\mu,\mu'}(t, k, \ell) = 
    \begin{cases}
      (x, 1) & \text{if $k = i$, $\ell = j$ and $\lambda_d(t) = \mu$} \\
      (1, x) & \text{if $k = i$, $\ell = j$ and $\lambda_d(t) = \mu'$} \\
      (1, 1) & \text{otherwise}
    \end{cases}
  \end{equation*}
  For each of the morphisms $b_{j,d}$ and $f_{j,d}$, the accepting set is
  $\os{x}$. For each $h_{i,j,\mu}$ and $v_{i,j,\mu}$, the accepting set is
  $\os{(1, 1), (x, x)}$. The accepting set for each $g_{i,j,d,\mu,\mu'}$ is
  $\os{(1, 1), (1, x), (x, 1)}$.

  For completeness, a correctness proof of the reduction is given in the
  appendix.
\end{proof}

The next objective is to obtain a stronger result in the case that $\vV$
contains some finite monoid which is not in $\vDS$.
Our proof is based on the well-known fact that direct products of $\abstar$ can
be used to encode computations of a Turing machine or runs of an automaton, an
idea which already appears in the proof of \cite[Theorem 4.9]{BeaudryMT92}.
To this end, we first describe classes of languages recognizable by such
direct products.

\begin{lemma}
  Let $\vV$ be a variety of finite monoids such that $\vV \not\subseteq
  \vDS$. Let $A$ be a finite alphabet and let $B, C, D, E, F$ be (possibly
  empty) pairwise disjoint subsets of $A$.
  Then, each of the languages $E^* B (D \union E)^*$, $(D \union E)^* C E^*$
  and $(E^* B (E \union F)^* C E^* \union E^* D E^*)^+$ is the preimage of an
  element of a monoid $M \in \vV$ of size $6$ under a morphism $h \colon A^*
  \to M$.
  \label{lem:lang1}
\end{lemma}
\begin{proof}
  Let $N$ be a monoid from $\vV \setminus \vDS$. By
  Proposition~\ref{prop:ds-char}, the monoid $\abstar$ is a divisor of the
  direct product $N \times N$ and since $\vV$ is closed under direct products
  and divisors, we have $\abstar \in \vV$. We let $M = \abstar$.

  For $E^*B(D \union E)^*$, consider the morphism $h \colon A^* \to M$ defined
  by $h(e) = 1$ for $e \in E$, $h(b) = b$ for $b \in B$, $h(d) = ab$ for all $d
  \in D$. All other letters are mapped to the zero element. By construction, we
  have $h^{-1}(b) = E^* B (D \union E)^*$. For $(D \union E)^* C E^*$, one can
  use a symmetrical construction.

  For $(E^* B (E \union F)^* C E^* \union E^* D E^*)^+$, we define $h \colon
  A^* \to M$ by $h(b) = a$ for all $b \in B$, $h(c) = b$ for $c \in C$, $h(d) =
  ab$ for $d \in D$, $h(f) = ba$ for $f \in F$ and $h(e) = 1$ for $e \in E$.
  Again, the remaining letters are mapped to $0$. The preimage of $ab$ is the
  desired language.
\end{proof}

\begin{lemma}
  Let $\vV$ be a variety of finite monoids such that $\vV \not\subseteq \vDS$.
  Let $A$ be a finite alphabet, let $n \in \N$ and let $A_1, \ldots, A_n$
  be pairwise disjoint subsets of $A$.
  Then the language $(A_1 \cdots A_n)^+$ can be written as an intersection
  of $n$ languages, each of which is the preimage of an element of a monoid $M
  \in \vV$ of size $6$ under a morphism $h \colon A^* \to M$.
  \label{lem:lang2}
\end{lemma}
\begin{proof}
  Let $B = A_1 \union \cdots \union A_n$. 
  For $1 \le i \le n-1$, we define the alphabet $D_i = B \setminus (A_i \union
  A_{i+1})$ and the language $L_i = (A_i A_{i+1} \union D_i)^+$. We also let
  $L_n = (A_1 D_n^* A_n)^+$ with $D_n = B \setminus (A_1 \cup A_n)$.
  By construction, we have $L_1 \intersect \cdots
  \intersect L_{n} = (A_1 \cdots A_n)^+$ and by Lemma~\ref{lem:lang1}, each
  of the languages $L_i$ is recognized by a monoid of size $6$.
\end{proof}

We are now able to state the second main theorem of this section.

\begin{theorem}
  Let $\vV$ be a variety of finite monoids such that $\vV \not\subseteq
  \vDS$. Then, the decision problem $\MonInt_1(\vV)$ is $\PSPACE$-complete.
  \label{thm:pspace-complete}
\end{theorem}
\begin{proof}
  Let $\T = (\lbl, T, n, f, b)$ be a tiling system.
  The objective is to construct a language $L$ which is non-empty if and only
  if there exists a valid $m$-tiling of $\T$ for some $m \in \N$.

  We may assume without loss of generality that $\lambda_w(t) \ne \lambda_e(t)$
  and $\lambda_s(t) \ne \lambda_n(t)$ for all tiles $t \in T$. If, for example,
  $\lambda_w(t) = \mu = \lambda_e(t)$ for a tile $t \in T$, we create a copy
  $\mu'$ of the label $\mu$ and replace every tile with $\lambda_w(t) = \mu$ by
  two copies. In one of the copies, we replace the west label with $\mu'$. We
  repeat this for all other directions and finally remove all tiles with
  $\lambda_w(t) = \lambda_e(t) \in \os{\mu, \mu'}$.

  We define an alphabet $A = T \times \os{0, 1, 2} \times \os{1, \dots, n}$.
  Intuitively, the letters of $A$ correspond to positions in a tiling. The
  first component describes the tile itself, the second component specifies
  whether the tile is in the first row, some intermediate row or in the bottom
  row and the third component specifies the column.
  For each $j \in \os{1, \dots, n}$ and $\mu \in \lbl$, let $C_j = T \times
  \os{0, 1, 2} \times \os{j}$ and $D_j = A \setminus C_j$ and
  \begin{align*}
    W_{\mu} & = \set{(t, i, j) \in A}{\lambda_w(t) = \mu, j > 1}, &
    N_{j, \mu} & = \set{(t, i, j) \in A}{\lambda_n(t) = \mu, i > 0}, \\
    E_{\mu} & = \set{(t, i, j) \in A}{\lambda_e(t) = \mu, j < n}, &
    S_{j, \mu} & = \set{(t, i, j) \in A}{\lambda_s(t) = \mu, i < 2}, \\
    X_{\mu} & = A \setminus (W_{\mu} \union E_{\mu}), &
    Y_{j, \mu} & = C_j \setminus (N_{j, \mu} \union S_{j, \mu}).
  \end{align*}
  Note that by our initial assumption, $W_\mu \intersect E_\mu = \emptyset$ and
  $N_{j,\mu} \intersect S_{j,\mu} = \emptyset$ for each $\mu \in \lbl$ and for
  $1 \le j \le n$.
  Let $F_j = \os{(t_j, 0, j)}$ and $B_j = \os{(u_j, 2, j)}$ where $t_j$ and
  $u_j$ are the tiles uniquely determined by $f = t_1 \cdots t_n$ and $b =
  u_1 \cdots u_n$.
  Let $\overline F_j = \set{(t, i, j) \in A}{i > 0}$ and
  $\overline B_j = \set{(t, i, j) \in A}{i < 2}$.
  We define
  \begin{align*}
    K =
      \left( \bigcap_{1 \le j \le n} D_j^* F_j (\overline F_j \union D_j)^* \right)
    & \intersect \left( \bigcap_{1 \le j \le n} (\overline B_j \union D_j)^* B_j D_j^* \right)
      \intersect \left( \bigcap_{\mu \in \lbl} (E_\mu W_\mu \union X_\mu)^+ \right) \\
    & \intersect \left( \bigcap_{\substack{\mu \in \lbl, \\ 1 \le j \le n}} (D_j^* S_{j, \mu} D_j^* N_{j, \mu} D_j^* \union D_j^* Y_{j, \mu} D_j^*)^+ \right).
  \end{align*}
  and $L = (C_1 \cdots C_n)^+ \intersect K$.
  By Lemma~\ref{lem:lang1} and Lemma~\ref{lem:lang2}, the language $L$ can be
  represented by a $\MonInt(\vV)$ instance with polynomially many morphisms to
  monoids of size $6$ from $\vV$ and with singleton accepting sets.
\end{proof}

\section{A Small Model Property for $\vDO$}
\label{sec:small-model}

The objective of this section is to prove the following result which states
that, within a non-empty intersection of languages recognized by monoids from $\vDO$,
there always exists a word with a small SLP representation.

\begin{theorem}
  For each $i \in \os{1, \dots, k}$, let $M_i \in \vDO$ and let $h_i \colon A^*
  \to M_i$ be a morphism. Let $w \in A^*$.
  Then there exists an SLP $S$ of size at most $p(N)$ with $h_i(S) = h_i(w)$
  for all $i \in \os{1, \dots, k}$ where $p \colon \R \to \R$ is some
  polynomial and $N = \abs{M_1} + \cdots + \abs{M_k}$.
  \label{thm:small-model}
\end{theorem}

Before diving into the proof of this result, we note that the theorem
immediately yields the following corollary:

\begin{corollary}
  $\MonInt(\vDO)$ is $\NP$-complete.
\end{corollary}
\begin{proof}
  In view of Theorem~\ref{thm:np-hard}, it suffices to describe an
  $\NP$-algorithm.
  The algorithm first non-deterministically guesses an SLP of polynomial size
  producing a word in the intersection of the given languages. It remains to
  check that the word represented in the SLP is indeed contained in each of the
  languages. To this end, we compute the image of the word represented by the
  SLP under each of the morphisms. Each such computation can be performed in
  time linear in the size of the SLP by computing the image of a variable $X$
  as soon as the images of all variables appearing on the right-hand side of
  $X$ are computed already, starting with minimal variables.
\end{proof}

\subsection{The Group Case}

We first take care of a special case, namely that each of the monoids is a
group.
In this case, one can use a variant of the \emph{Schreier-Sims
algorithm}~\cite{Sims1967,FurstHopcroftLuks80} to obtain a compressed
representative. To keep the paper self-contained, we give the full algorithm
alongside with a correctness proof.

Our setting is as follows: the input are groups $G_1$, \ldots, $G_k$
which are, without loss of generality, assumed to be pairwise disjoint, and
morphisms $h_i \colon A^* \to G_i$ with $i \in \os{1, \dots, k}$.
We let $G = G_1 \union \cdots \union G_k$ and $N = \abs{G}$. Note
that $G$ is considered as a set; it does not form a group unless $k = 1$.
However, for each $g \in G$, we interpret powers $g^i$ in the corresponding
group $G_i$ with $g \in G_i$.
We let $\omega = N!$ so that, for each $g \in G$, the element $g^\omega$ is the
identity.\footnote{One could also choose $\omega = \lcm\os{\abs{G_1}, \dots, \abs{G_k}}$ but for the analysis, it does not matter, since
$N!$ is sufficiently small.}

\begin{algorithm}[h!]
  \caption{The \textsc{sift} procedure}
  \begin{algorithmic}
    \Procedure{sift}{$\alpha$}
    \State $R_0 \gets \varepsilon$
    \For{$i \in \os{1, \dots, k}$}
      \State $S_i \gets R_{i-1}^{\omega-1} \alpha$
      \LineIf{$T[h_i(S_i)] = \varepsilon$}{$T[h_i(S_i)] \gets S_i$}
      \State $R_i \gets R_{i-1} T[h_i(S_i)]$
    \EndFor
    \State \Return $R_k$
    \EndProcedure
  \end{algorithmic}
\end{algorithm}

The algorithm maintains a table $T \colon G \to (A \union V)^*$ as an internal
data structure, where the set of variables $V$ is extended as needed and the
table entries $T[g]$ can be considered variables themselves.
The $\Sift$ procedure expects a parameter $\alpha \in (V \union A)^*$ and tries
to find a short representation of $\val(\alpha)$, using only entries from the
table unless it comes across an empty table entry, in which case it uses
$\alpha$ to fill the missing table entry itself.
When a table entry is assigned a word with a factor of the form $X^{\omega-1}$,
this factor is stored in a compressed form by using the technique from
Lemma~\ref{lem:slp} and adding new variables as needed. Thus, a factor
$X^{\omega-1}$ only requires $4 \log(\omega-1) \le 4 \log(N!) \le 4 N \log(N)$
additional space.

\begin{algorithm}[h!]
  \caption{Initialization of the compression algorithm for groups}
  \begin{algorithmic}
    \Procedure{init}{}
      \LineForAll{$g \in G$}{$T[g] \gets \varepsilon$}
      \State $c \gets 0$
      \Repeat
        \State $c_p \gets c$
        \ForAll{$g_1 \in G_1, \dots g_k \in G_k, a \in A$}
          \State $\Call{sift}{T[g_1] \cdots T[g_k] a}$
        \EndFor
        \State $c \gets \abs{\set{g \in G}{T[g] \ne \varepsilon}}$
      \Until{$c = c_p$}
    \EndProcedure
  \end{algorithmic}
  \label{alg:main}
\end{algorithm}

Before the $\Sift$ procedure is used for compression, the table needs to be
initialized.
To this end, the $\Init$ routine fills the table with short representatives
such that future $\Sift$ invocations never run into empty table entries again.
Let us first prove several invariants of the $\Sift$ procedure.

\begin{lemma}
  For each $i \in \os{1, \dots, k}$ and $g \in G_i$, we have $T[g] =
  \varepsilon$ or $h_i(T[g]) = g$.
  \label{lem:table}
\end{lemma}

\begin{proof}
  Suppose that $T[g] \ne \varepsilon$. Then, in some round of the $\Sift$
  procedure, we have $h_i(S_i) = g$ and $T[h_i(S_i)]$ is assigned the SLP $S_i$
  (and never modified again). Therefore, $h_i(T[g]) = h_i(T[h_i(S_i)]) =
  h_i(S_i) = g$.
\end{proof}

\begin{lemma}
  After round $i$ of the $\Sift$ procedure, we have $h_i(R_i) = h_i(\alpha)$.
  \label{lem:step}
\end{lemma}
\begin{proof}
  By the definition of $R_i$, we have $h_i(R_i) = h_i(R_{i-1}
  T[h_i(S_i)])$ which is the same as $h_i(R_{i-1} S_i)$ by
  Lemma~\ref{lem:table}. Plugging in the definition of $S_i$ yields
  $h_i(R_{i-1} R_{i-1}^{\omega-1} \alpha) = h_i(\alpha)$ where the latter
  equality holds since $G_i$ is a group.
\end{proof}

\begin{lemma}
  For $1 \le i < j \le k$ and for all $g \in G_j$, we have $h_i(T[g]) = 1$.
  \label{lem:identity}
\end{lemma}
\begin{proof}
  Consider the invocation of the $\Sift$ procedure where $T[g]$ is defined.
  In round $j$ of this invocation, the entry $T[g]$ is assigned some SLP $S_j$
  with $h_j(S_j) = g$.

  Therefore, $h_i(T[g]) = h_i(S_j) = h_i(R_{j-1}^{\omega-1} \alpha)$. Expanding
  $R_{j-1}$ yields
  \begin{equation*}
    h_i(R_{j-1}) =
    h_i \big(\left( \prod_{r=1}^{j-1} T[h_r(S_r)] \right) \big) =
    h_i \big(\left( \prod_{r=1}^{i} T[h_r(S_r)] \right) \big) =
    h_i(R_{i})
  \end{equation*}
  where the second equality follows by induction. Therefore, $h_i(T[g]) =
  h_i(R_i^{\omega-1} \alpha)$ which is the same as $h_i(\alpha^{\omega-1}
  \alpha) = 1$ by Lemma~\ref{lem:step}.
\end{proof}

\begin{lemma}
  After round $j$, we have $h_i(R_k) = h_i(\alpha)$ for all $i \in
  \os{1, \dots, j}$.
  \label{lem:sift-return}
\end{lemma}
\begin{proof}
  Using the expansion of $R_k$ and Lemma~\ref{lem:identity}, we obtain the
  sequence of equalities
  \begin{equation*}
    h_i(R_k) =
    h_i \big(\left( \prod_{r=1}^{k} T[h_r(S_r)] \right) \big) =
    h_i \big(\left( \prod_{r=1}^{i} T[h_r(S_r)] \right) \big) =
    h_i(R_i).
  \end{equation*}
  The statement now follows immediately from Lemma~\ref{lem:step}.
\end{proof}

\begin{theorem}
  For each $i \in \os{1, \dots, k}$, let $G_i$ be a finite group and let $h_i
  \colon A^* \to G_i$ be a morphism. Let $w \in A^*$.
  Then there exists an SLP $S$ of size at most $p(N)$ with $h_i(S) = h_i(w)$
  for all $i \in \os{1, \dots, k}$ where $p \colon \R \to \R$ is some
  polynomial and $N = \abs{G_1} + \cdots + \abs{G_k}$.
  \label{thm:groups}
\end{theorem}
\begin{proof}
  We claim that the SLP $S$ constructed when calling $\Init$, followed by
  $\Sift$ with parameter $w$ satisfies the properties above.
  By Lemma~\ref{lem:sift-return}, we have $h_i(S) = h_i(w)$ for all $i \in
  \os{1, \dots, k}$.
  Moreover, when the initialization routine returns, the table entries contain
  SLP of polynomially bounded size. We now claim that any subsequent executions
  of the $\Sift$ procedure will not define any new table entries, no matter
  which SLP is passed as a parameter. In particular, running $\Sift(w)$ yields
  an SLP that only uses already existing table entries.

  To prove the claim, assume, for the sake of contradiction, that there exists
  some word $v$ such that some new table entry $T[g]$ is defined during
  $\Sift(v)$. We choose $v$ such that it is a word of minimal length
  satisfying this condition. This means that we can factorize $v = v'a$ with $a
  \in A$ such that all table entries are defined when calling $\Sift(v')$.
  Let $T[g_1] \cdots T[g_k]$ be the return value of $\Sift(v')$.
  Then $\Sift(T[g_1] \cdots T[g_k] a)$ is called during the
  initialization process and because $h_i(T[g_1] \cdots T[g_k] a)
  = h_i(v)$ for all $1 \le i \le k$, the sequence of $S_i$ during the
  execution of $\Sift(T[g_1] \cdots T[g_k] a)$ is the same as in
  $\Sift(v)$ which means that all table entries accessed during $\Sift(v)$ are
  defined.
\end{proof}

\subsection{The General Case}

For the general case, where each of the monoids is in $\vDO$ but not
necessarily a group, we use combinatorial properties of languages recognized by
monoids from $\vDO$ to reduce the problem to the group case.
The following lemmas are an essential ingredient of this reduction.

\begin{lemma}
  Let $h \colon A^* \to M$ be a morphism to a finite monoid $M \in \vDS$.
  Let $u, v \in A^*$ such that $h(v) \in E(M)$ and $\alp(u) \subseteq \alp(v)$.
  Then $h(v) \Jle h(u)$.
  \label{lem:ds}
\end{lemma}
\begin{proof}
  Let $u = a_1 \cdots a_\ell$ with $a_i \in A$ for $1 \le i \le \ell$.
  Since $a_i \in \alp(v)$ for each $i \in \os{1, \dots, \ell}$, we have $h(v)
  \Jle h(a_i)$. By Proposition~\ref{prop:ds-char}, the set $\set{x \in M}{h(v)
  \Jle x}$ is a submonoid of $M$ which means that $h(v) \Jle h(a_1) 
  \cdots h(a_\ell) = h(u)$, thereby proving the claim.
\end{proof}

\begin{lemma}
  Let $M \in \vDO$ and let $e, f, g \in E(M)$ with $e \Jeq f \Jle g$. Then $egf
  = ef$.
  \label{lem:do-pre}
\end{lemma}
\begin{proof}
  First note that $(fg)^\omega \Jeq (fgf)^\omega \Jeq (gf)^\omega$. Since $M
  \in \vDS$, we have $(fgf)^\omega = f$ and since $M \in \vDO$, we have
  $(fg)^\omega = (fg)^\omega (gf)^\omega (fg)^\omega = (fg)^{\omega-1}$.
  Together, this yields, $fgf = (fgf)^\omega gf = (fg)^\omega fgf =
  (fg)^{\omega-1} fgf = (fg)^\omega f = (fgf)^\omega = f$, thus $gf \in E(M)$.
  By Proposition~\ref{prop:ds-char}, we obtain $gf \Jeq e$.
  Therefore, $egf = eg(fef) = (e \ms gf \ms e) f = ef$.
\end{proof}

\begin{lemma}
  Let $M \in \vDO$, let $e, f, g \in E(M)$ and let $x, y \in M$ such that $e
  \Jeq f \Jle g, x, y$. Then $exgyf = exyf$.
  \label{lem:do}
\end{lemma}
\begin{proof}
  Since $M \in \vDS$, we have $ex = (exe)^\omega x = ex (ex)^\omega$ and $yf =
  y (fyf)^\omega = (yf)^\omega yf$.
  Note that $(ex)^\omega \Jeq e \Jeq f \Jeq (yf)^\omega$ by
  Proposition~\ref{prop:ds-char} and thus, Lemma~\ref{lem:do-pre} yields
  $(ex)^\omega g (yf)^\omega = (ex)^\omega (yf)^\omega$.
  Finally, combining all the equalities, we obtain the desired statement $exgyf
  = ex (ex)^\omega g (yf)^\omega yf = ex (ex)^\omega (yf)^\omega yf = exyf$.
\end{proof}

For the remainder of this section, let $M_1$, \ldots, $M_k \in \vDO$ be
finite monoids and let $h_i \colon A^* \to M_i$ be morphisms. We let $N =
\abs{M_1} + \cdots + \abs{M_k}$.
The occurrence of a word $u$ in $puq$ is called \emph{isolated} if for each $i \in \os{1, \dots,
k}$, there exist words $v_i, w_i \in A^*$ such that
\begin{equation*}
  \alp(v_i) = \alp(w_i) \supseteq \alp(u), \quad
  h_i(p v_i) = h_i(p) \quad \text{~and~} \quad
  h_i(w_i q) = h_i(q).
\end{equation*}

Let $w = a_1 u_1 a_2 \cdots u_{\ell-1} a_\ell$ be a factorization of $w$
with $a_j \in A$ and $u_j \in A^*$ for all $j \in \os{1, \dots, \ell}$. Let
$p_j = a_1 u_1 a_2 \cdots u_{j-1} a_j$ and $q_j = a_{j+1} u_{j+1} \cdots a_{\ell-1} u_{\ell-1}
a_\ell$.
The factorization $w = a_1 u_1 a_2 \cdots u_{\ell-1} a_\ell$ is called
\emph{piecewise isolating}
%mkk{piecewise isolated}
 if, for each $j \in \os{1, \dots, \ell-1}$, the
occurrence of $u_j$ in $w = p_j u_j q_j$ is isolated.
The value $\ell$ is the \emph{length} of this factorization.

\begin{lemma}
  Every word $w \in A^*$ admits a piecewise isolating factorization of length
  at most $N^2$.
\end{lemma}
\begin{proof}
  Let $w = b_1 \cdots b_m$ where $b_r \in A$ for $1 \le r \le m$. To each
  position $r \in \os{1, \dots, m}$, we assign a set $C_r = \set{(h_i(b_1
  \cdots b_s), h_i(b_{s+1} \cdots b_m))}{1 \le i \le k, 1 \le s \le
  r}$.
  Note that by definition, we have $C_r \subseteq C_{r+1}$. Let $r_1,
  \dots, r_\ell \in \N$ such that $r_1 = 1$, $r_\ell= m$ and $C_{r_{j-1}} =
  C_{r_j - 1} \subsetneq C_{r_j}$ for all $j \in \os{2, \dots, \ell}$. Let $a_j
  = b_{r_j}$ and let $u_j = b_{r_j+1} \cdots b_{r_{j+1}-1}$ for all
  $j \in \os{1, \dots, \ell}$.

  Now, for $j \in \os{1, \dots, \ell}$ and $i \in \os{1, \dots, k}$, let $t(j,i)$ be the smallest index $g$ such that $(h_i(a_1 u_1 \cdots
  a_{j} u_{j}), h_i(a_{j+1} u_{j+1} \cdots a_{\ell-1} u_{\ell-1} a_\ell)) \in C_{r_g}$, \ie the prefix of length $r_{t(j,i)}$ of $w$ is the shortest
  prefix $p$ such that $w = pq$ for some $q \in A^*$ and the image of $p$ under
  $h_i$ is $h_i(a_1 u_1 \cdots a_{j} u_{j})$ and the image of $q$ is
  $h_i(a_{j+1} u_{j+1} \cdots a_{\ell-1} u_{\ell-1} a_\ell)$.
  Note that $t(j,i) \le j$ and, by choice of $t(j,i)$, we have
  \begin{align}
    h_i(a_1 u_1 \cdots a_{j} u_{j}) & = h_i(a_1 u_1 a_2 \cdots u_{t(j,i)-1} a_{t(j,i)}) \text{~and} \label{eq1} \\
    h_i(a_{j+1} u_{j+1} \cdots a_{\ell-1} u_{\ell-1} a_\ell) & = h_i(u_{t(j,i)} a_{t(j,i)+1} \cdots u_{\ell-1} a_\ell). \label{eq2}
  \end{align}
  Let $w_{ji}=u_{t(j,i)} a_{t(j,i)+1} \cdots u_{j-1} a_j u_j$ 
  and  let $v_{ji} = u_j u_{t(j,i)} a_{t(j,i)+1} \cdots u_{j-1} a_j$.
  In the special case $t(j,i) = j$, we obtain $w_{ji} = v_{ji} = u_j$.
  
  \noindent
  For $p_{j} = a_1 u_1 a_2 \cdots u_{j-1} a_j$ and $q_{j} = a_{j+1} u_{j+1}
  \cdots a_{\ell-1} u_{\ell-1} a_\ell$, 
  equation (\ref{eq1}) implies
  \begin{align*}
    h_i(p_{j} v_{ji}) & = h_i(a_1 u_1  \cdots a_j u_j) \cdot h_i(u_{t(j,i)} a_{t(j,i)+1} \cdots u_{j-1} a_j) \\
               & = h_i(a_1 u_1 a_2 \cdots u_{t(j,i)-1} a_{t(j,i)}) \cdot h_i(u_{t(j,i)} a_{t(j,i)+1} \cdots u_{j-1} a_j) = h_i(p_j)
  \end{align*}
  and, similarly, equation (\ref{eq2}) yields $h_i(w_{ji} q_{j}) = h_i(q_{j})$.
  Since $u_j$ is a suffix of $w_{ji}$ and since $v_{ji}$ can be obtained by
  rotating $w_{ji}$ cyclically, we have $\alp(v_{ji}) = \alp(w_{ji}) \supseteq
  \alp(u_j)$.
  The bound on $\ell$ follows from the fact that $C_{r_1} \subsetneq \cdots
  \subsetneq C_{r_\ell} \subseteq \bigunion_{i=1}^k {M_i \times M_i}$.
\end{proof}

The lemma above suggests that it is sufficient to construct SLPs for isolated
occurrences. Thus, let now $u \in A^*$ be an isolated occurrence of $w = puq$,
and let $B = \alp(u)$.
For each $i \in \os{1, \dots, k}$, we define an equivalence relation $\equiv_i$
on the submonoid $T_i = h_i(B^*)$ of $M_i$ by $m \equiv_i n$ if and only if
$h_i(p) \ms xmy \ms h_i(q) = h_i(p) \ms xny \ms h_i(q)$ for all $x, y \in T_i$.
It is easy to check that this relation is a congruence.
Moreover, for all $u, v \in B^*$ with $h_i(u) \equiv_i h_i(v)$, we have
$h_i(puq) = h_i(pvq)$.
Another fundamental property of $\equiv_i$ is captured in the following lemma.

\begin{lemma}
  For each $i \in \os{1, \dots, k}$, the quotient ${T_i} / {\equiv_i}$ is a group.
  \label{lem:groupred}
\end{lemma}
\begin{proof}
  Let $\omega = N!$ and let $m \in T_i$ be an arbitrary element. It suffices to
  show that $m^\omega \equiv_i 1$, \ie{}for all $x, y \in T_i$, we have $h_i(p)
  \ms xm^\omega y \ms h_i(q) = h_i(p) \ms xy \ms h_i(q)$.

  Let $v_i, w_i \in A^*$ as in the definition of isolated occurrences and let
  $e = h(v_i^\omega)$ and $f = h(w_i^\omega)$.
  Note that $h_i(p v_i) = h_i(p)$ implies $h_i(p) e =
  h_i(p)$. Analogously, we have $f h_i(q) = h_i(q)$.
  Since $B$ is contained in $\alp(v_i) = \alp(w_i)$ and since $m, x, y \in T_i
  = h_i(B^*)$, we have $e \Jeq f \Jle m^\omega, x, y$ by Lemma~\ref{lem:ds}.
  Therefore,
  \begin{equation*}
    h_i(p) \ms x m^\omega y \ms h_j(q) =
    h_i(p) \ms ex m^\omega yf \ms h_i(q) =
    h_i(p) \ms exyf \ms h_i(q) =
    h_i(p) \ms xy \ms h_i(q)
  \end{equation*}
  where the second equality uses Lemma~\ref{lem:do}.
\end{proof}

We now return to the proof of the main theorem of this section.

\begin{proof}[Proof of Theorem~\ref{thm:small-model}]
  By considering a piecewise isolating factorization of $w$, it suffices to
  show that if $u$ is an isolated occurrence in $w=puq$, then there exists
  an SLP $S$ of polynomial size with $h_i(pSq) = h_i(puq)$ for all $i \in
  \os{1, \dots, k}$. Combining the letters $a_i$ and the SLPs for the isolated
  occurrences in the piecewise isolating factorization, we obtain the SLP for
  $w$.

  Let again $B = \alp(u)$.
  To obtain a polynomial-size SLP $S$ with $h_i(pSq) = h_i(puq)$ for all $i \in
  \os{1, \dots, k}$, we consider the morphisms $\psi_i \colon B^* \to
  {T_i}/{\equiv_i}$ defined by $\psi_i(v) = {[h_i(v)]}_{\equiv_i}$, \ie{}each
  word $v$ is mapped to the equivalence class of $h_i(v)$ with respect to
  $\equiv_i$.
  Note that $\abs{{T_i}/{\equiv_i}} \le \abs{T_i} \le \abs{M_i}$ for $1 \le i
  \le k$ and by Lemma~\ref{lem:groupred}, each of the monoids
  ${T_i}/{\equiv_i}$ is a group.
  By Theorem~\ref{thm:groups}, there exists a polynomial-size SLP $S$ with
  $\psi_i(S) = \psi_i(u)$ for all $i \in \os{1, \dots, k}$ and, by
  the definition of $\equiv_i$, we obtain $h_i(pSq) = h_i(puq)$, as desired.
\end{proof}

\section{Summary and Outlook}

We investigated the complexity of the intersection problem for finite monoids,
showing that the problem is $\NP$-complete for varieties contained in $\vDO$
and $\PSPACE$-complete for varieties not contained within $\vDS$.
To obtain a dichotomy result, one needs to investigate the complexity of the
problem when monoids from $\vDS \setminus \vDO$ are part of the input. Using
techniques similar to those in Section~\ref{sec:small-model}, we were able to
show that for a subset of this class, the problem remains $\NP$-complete and
thus, we conjecture that the problem is $\NP$-complete whenever $\vV \subseteq
\vDS$.
The fact that $\vDS \setminus \vDO$ have not been studied and understood well
enough from a language-theoretic perspective makes the problem of classifying
the complexity of these monoids challenging but, at the same time, an
interesting object for further research.

\newcommand{\Ju}{Ju}\newcommand{\Ph}{Ph}\newcommand{\Th}{Th}\newcommand{\Ch}{Ch}\newcommand{\Yu}{Yu}\newcommand{\Zh}{Zh}\newcommand{\St}{St}\newcommand{\curlybraces}[1]{\{#1\}}

\clearpage
\appendix

\section{Correctness of the Reduction in Theorem~\ref{thm:np-hard}}

We skipped the correctness proof for the reduction presented in the proof of
Theorem~\ref{thm:np-hard}. Since correctness may not be entirely obvious at
first sight, we give the missing arguments here. We assume without loss of
generality that $\lbl$ contains at least two different labels.

First, note that, if there exists an $n$-tiling $\tau \colon \os{1, \dots, n}
\times \os{1, \dots, n} \to T$, then the word obtained by concatenating all
letters $(\tau(i, j), i, j)$ with $i, j \in \os{1, \dots, n}$ yields a word
contained in each of the languages recognized by the constructed morphisms.

Conversely, suppose that there exists a word $w \in A^*$ contained in each of
the languages recognized by the constructed morphisms, \ie{}the following
properties hold:

\begin{enumerate}
  \item $f_{j,d}(w) = b_{j,d}(w) = x$ for $1 \le j \le n$ and $d \in \os{w, e, s, n}$,
  \item $h_{i,j,\mu}(w) \in \os{(1, 1), (x, x)}$ for $1 \le i \le n$ and $1 \le j \le n-1$ and $\mu \in \lbl$, \label{enum:aaa}
  \item $v_{i,j,\mu}(w) \in \os{(1, 1), (x, x)}$ for $1 \le i \le n-1$ and $1 \le j \le n$ and $\mu \in \lbl$,
  \item $g_{i,j,d,\mu,\mu'}(w) \in \os{(1, 1), (1, x), (x, 1)}$ for $1 \le i, j \le n$ and $d \in \os{w, e, s, n}$ and $\mu \ne \mu' \in \lbl$.
\end{enumerate}

For $i, j \in \os{1, \dots, n}$ and $d \in \os{w, e, s, n}$ and $\mu \in \lbl$,
we say that position $(i, j)$ is \emph{$\mu$-labelled in direction $d$} if
$g_{i,j,d,\mu,\mu'} \in \os{x} \times M$ for some $\mu' \ne \mu$. Note that if
$\mu \ne \mu'$, a position $(i, j)$ cannot be both $\mu$-labelled and
$\mu'$-labelled in the same direction $d$, since otherwise, we would have
$g_{i,j,d,\mu,\mu'} = (x, x)$.
We now define a tiling $\tau \colon \os{1, \dots, n} \times \os{1, \dots, n}$
by setting $\tau(i, j) = (\mu_w, \mu_e, \mu_s, \mu_n)$ if position $(i, j)$ is
$\mu_d$-labelled in direction $d$ for all $d \in \os{w, e, s, n}$.
It remains to show that $\tau$ is a valid $n$-tiling of $\T$.

Let $f = t_1 \cdots t_n$ and let $b = u_1 \cdots u_n$.
Consider some arbitrary $j \in \os{1, \dots, n}$.
First of all, since $f_{j,d}(w) = x$ for each direction $d \in \os{w, e, s,
n}$, the position $(1, j)$ is $\lambda_d(t_j)$-labeled in direction $d$ and,
equivalently, the position $(n, j)$ is $\lambda_d(u_j)$-labelled in direction
$d$ since $b_{j,d}(w) = x$. This means that $\tau(1, j) = t_j$ and $\tau(n, j)
= u_j$.

Suppose now, for the sake of contradiction, that a position $(i, j)$ is
$\mu$-labeled in direction $e$ and that position $(i, j+1)$ is $\mu'$-labeled
in direction $w$ for $\mu, \mu' \in \lbl$ with $\mu \ne \mu'$. Then, by the
remark above, position $(i, j+1)$ cannot be $\mu$-labeled in direction $w$
which yields $h_{i,j,\mu} = (x, y)$ for some $y \in M \setminus \os{x}$,
contradicting property~\ref{enum:aaa} above.
Using the same arguments, one can show that if $(i, j)$ is
$\mu$-labeled in direction $s$, then $(i+1, j)$ is also $\mu$-labeled in
direction $n$.
Thus, we have $\lambda_e(\tau(i, j)) = \lambda_w(\tau(i, j+1))$ and
$\lambda_s(\tau(i, j)) = \lambda_n(\tau(i+1,j))$.
\qed

\end{document}